\documentclass[pra,twocolumn,showpacs,nofootinbib,superscriptaddress,notitlepage]{revtex4-1}
\usepackage{amsmath,empheq}
\usepackage{amssymb,bm}
\usepackage{amsthm,color,xcolor,dsfont}
\usepackage{graphicx} 
\usepackage{xcolor}
\usepackage{epstopdf}
\usepackage[colorlinks=true, hyperindex, breaklinks, linkcolor=blue, urlcolor=blue, citecolor=blue]{hyperref} 
\usepackage{ulem}
\usepackage{caption}
\usepackage[labelformat=simple]{subcaption}

\usepackage{cleveref}
\usepackage{mathtools}
\usepackage{appendix}

\input{Qcircuit}

\DeclarePairedDelimiter\ceil{\lceil}{\rceil}

\newcommand{\e}{\epsilon}

 \newtheorem{lemma}{Lemma}

\begin{document}

\title{Thresholds for universal concatenated quantum codes}

\author{Christopher Chamberland}
\email{c6chambe@uwaterloo.ca}
\affiliation{
    Institute for Quantum Computing and Department of Physics and Astronomy,
    University of Waterloo,
    Waterloo, Ontario, N2L 3G1, Canada
    }
\author{Tomas Jochym-O'Connor}
\email{trjochym@uwaterloo.ca}
\affiliation{
    Institute for Quantum Computing and Department of Physics and Astronomy,
    University of Waterloo,
    Waterloo, Ontario, N2L 3G1, Canada
    }
\author{Raymond Laflamme}
\affiliation{
    Institute for Quantum Computing and Department of Physics and Astronomy,
    University of Waterloo,
    Waterloo, Ontario, N2L 3G1, Canada
    }
\affiliation{Perimeter Institute, Waterloo, Ontario, N2L~2Y5, Canada}
\affiliation{Canadian Institute For Advanced Research, Toronto, Ontario, M5G~1Z8, Canada}

\begin{abstract}
Quantum error correction and fault-tolerance make it possible to perform quantum computations in the presence of imprecision and imperfections of realistic devices. An important question is to find the noise rate at which errors can be arbitrarily suppressed. By concatenating the 7-qubit Steane and 15-qubit Reed-Muller codes, the 105-qubit code enables a universal set of fault-tolerant gates despite not all of them being transversal. Importantly, the CNOT gate remains transversal in both codes, and as such has increased error protection relative to the other single qubit logical gates. We show that while the level-1 pseudo-threshold for the concatenated scheme is limited by the logical Hadamard, the error suppression of the logical CNOT gates allows for the asymptotic threshold to increase by orders of magnitude at higher levels. We establish a lower bound of~$1.28 \times 10^{-3}$ for the asymptotic threshold of this code which is competitive with known concatenated models and does not rely on ancillary magic state preparation for universal computation.
\end{abstract}


\maketitle

\section{Introduction}
Quantum computers have the potential to greatly enhance the efficiency of certain computational problems. However, they rely on the storage and manipulation of many quantum systems in superposition, and it is this careful juxtaposition of storage and manipulation of the quantum states that renders their development to be so difficult. Namely, by making individual quantum systems easily accessible to control often leads to increased external noise. Suppressing noise in a scalable manner is thus a necessary requirement for any quantum computing architecture, promoting the need for quantum error correction and fault-tolerance. 

Quantum error correcting codes come in many different forms, yet the key to any error correcting scheme is the establishment of a fault-tolerance threshold~\cite{Shor96, AB97, Preskill98, KLZ98}. Concatenated codes have played a key role in determining these threshold rates due to their ability to iteratively suppress errors by increasing levels of concatenation. Along these lines, Aliferis, Gottesman, and Preskill established a rigorous lower bound for the fault-tolerance threshold of concatenated codes by introducing a technique called malignant set counting~\cite{AGP06}. Paetznick and Reichardt used this method to establish a circuit level noise threshold for the 23-qubit Golay~code under physical depolarizing noise, obtaining a threshold error rate of~$1.32 \times 10^{-3}$.  

One of the most prominent methods for implementing a logical fault-tolerant gate is by implementing the gate transversally, that is by applying individual physical gates to each of the qubits composing the logical qubit. However, as shown by Eastin and Knill, the set of transversal gates for a given code generates a finite group, and therefore cannot be universal for quantum computation~\cite{EK09}. In order to circumvent this fundamental restriction and potentially reduce the qubit overhead seen in magic state distillation~\cite{BK05}, many recent fault-tolerant proposals for universal logical quantum computation have focused on code conversion and gauge fixing~\cite{PR13,ADP14, Bombin14, BC15, JB16, JBH15}. In this work, we study a parallel construction for universal fault-tolerant quantum computation through the concatenation of two error correcting codes~\cite{JL14}. The idea behind this scheme is to protect the gate that is not transversal for one code through its implementation using transversal gates in the other code. The concatenation scheme provides a dual protection for the purposes of universal fault-tolerant quantum logic.  

In this work, we establish a lower bound on the fault-tolerant threshold for the 105-qubit universal concatenated code under depolarizing noise. We show that the dual protection coming from the concatenation of two different error correcting codes provides more than just minimal fault-tolerant protection, it serves as a means for logical error suppression at the second level of concatenation~(and above). We believe that this provides new insights in the development of quantum error correcting codes, and emphasizes an important principle: to logically protect the quantum gates that are most present in the fault-tolerant architecture.

\section{Concatenated 105-qubit scheme} \label{Concatenated 105-qubit scheme}

\begin{figure}[h]
\centering
\begin{subfigure}{0.4\textwidth}
\includegraphics{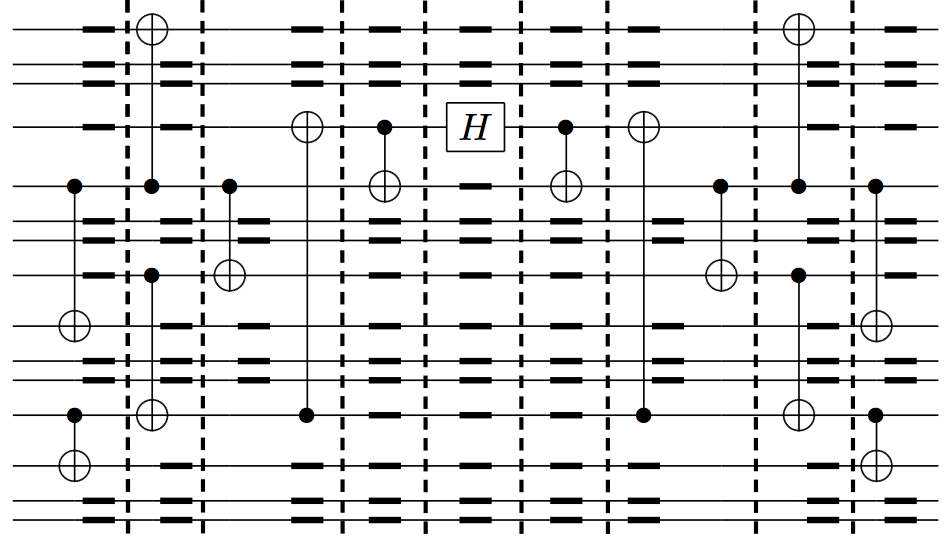}
\caption{}
\label{fig:HadCircuit}
\end{subfigure}
\begin{subfigure}{0.4\textwidth}
\begin{align*}
\Qcircuit @C=2em @R=1.4em {
& \gate{\text{LEC}} & \gate{G} & \gate{\text{TEC}} & \qw
\gategroup{1}{2}{1}{4}{0.7em}{--}
}
\end{align*}
\caption{}
\label{fig:exRecCircuit}
\end{subfigure}
\caption{\subref{fig:HadCircuit}~Logical Hadamard~$H$ circuit for $[[15,1,3]]$~Reed-Muller code. The bold dark lines represent resting qubits subject to storage errors. The dotted vertical lines are used to separate the time steps for which gates are applied in parallel. Logical~$H$ for the 105-qubit code is implemented fault-tolerantly by applying each non-fault-tolerant logical~$H$ gates in parallel. \subref{fig:exRecCircuit}~Extended rectangle consisting of leading and trailing error correcting circuits implementing the desired logical gate $G$.}
\label{fig:Circuits}
\end{figure}

We begin by briefly reviewing the 105-qubit concatenated scheme for universal fault-tolerant logical gates~\cite{JL14}. The logical information is encoded through the concatenation of an outer and inner quantum code, that is each logical qubit of the outer code is in turn encoded into the code of the inner code. In the 105-qubit code, the outer code is the 7-qubit Steane code, which has the properties of having transversal Clifford operations. The inner code is the 15-qubit Reed-Muller code, which contains CNOT and $T = \text{diag}(1, e^{i\pi/4})$ as its transversal gates. The gate set generated by Clifford~+~$T$ is universal for quantum computation~\cite{BBC+95}. The overall code is a~$[[105,1,9]]$ code encoding a single logical qubit in 105 qubits with distance~9. Since CNOT and the phase gate~$S = T^2$ is transversal for both codes, both gates will remain transversal when the two codes are concatenated. Logical Hadamard~$H$ is obtained by applying a non-transversal logical Hadamard to each of the encoded 15-qubit codeblocks. Although not fault-tolerant on each 15-qubit codeblock with a single error potentially leading to a logical error, due to the protection of the 7~qubit code, a single error will never result in a global logical fault and will remain correctable. Figure~\ref{fig:HadCircuit} summarizes the application of the logical~$H$ gate on a 15-qubit codeblock. Note that the circuit construction was optimized using only 14 CNOT gates with a circuit depth of 9 time steps. It might be possible to find a circuit using fewer gates and a better depth. 

Logical~$T$ is constructed by choosing a sequence of logical CNOT and $T$~gates to be implemented on the 7-qubit outer code (see Fig.~\ref{fig:LogicalTgateCircuit} in the Supplementary Material). While this operation is not fault-tolerant on the outer code as errors can be spread between codeblocks, the underlying logical gates are transversal on each of the 15-qubit codeblocks. As such, any single error may result in multiple single errors spread across different codeblocks, and will remain correctable by the error correction of each of the 15-qubit codeblocks. 

\section{Fault-tolerance threshold theorem}

The key property of fault-tolerant architectures is the presence of an \emph{asymptotic threshold}. For concatenated coding schemes, the asymptotic threshold is the physical error rate $p_{th}$ such that for physical error rates~$p < p_{th}$ the logical error rate can be made arbitrarily small for sufficiently large number of concatenation levels (and the overall time/space resource overhead scales as~$\mathcal{O}(\text{poly}(\log{(A/\epsilon)})A)$, where $A$ would be the required resources for a noiseless circuit).

All currently known fault-tolerant schemes for quantum logic require active error correction between logical gates. Error correction steps are interleaved between the implementation of various fault-tolerant gates. In this study, fault-tolerant syndrome measurement and error correction is implemented using Steane's method~\cite{Steane97}. At a given concatenation level, each component of the logical circuit (gates and error detection/measurement) will be themselves composed of many operations from the previous level of concatenation. These components include state preparation and measurement, logical gates and memory locations. We consider a depolarizing model for each physical location (level-0) in the circuit. Depolarizing noise is modelled in a similar manner to that of Paetznick and Reichardt in their study of the 23-qubit Golay code~\cite{PR12}. Each single qubit gate (including resting qubits) undergoes Pauli noise with probability~$p/4$ for each Pauli operation, and each two qubit gate undergoes two-qubit Pauli noise with probability~$p/16$ for each non-trivial two-qubit Pauli. Under this noise, state preparation in the stabilizer $Z \ (X)$~basis is flipped from $\ket{0} \ (\ket{+})$ to $\ket{1} \ (\ket{-})$ with probability~$p/4$. Similarly, measurement in the stabilizer $Z \ (X)$~basis is flipped with probability~$p/4$. 

As first proposed by Aliferis~\emph{et al.}~\cite{AGP06} we analyze logical gates by considering the whole as an extended rectangle (\emph{exRec}), that is the logical gate itself along with its leading~(LEC) and trailing~(TEC) error correction circuits (see Figure~\ref{fig:exRecCircuit}). In order to characterize the rate at which logical errors occur, we define \emph{malignant} error events. Let $\ket{\psi_1}$ be a single or two-qubit logical state obtained by applying ideal decoders immediately after the LEC circuit and $\ket{\psi_2}$ the state obtained by applying ideal decoders immediately after the TEC. We define the event $\text{mal}_{E}$ as $\ket{\psi_2} = EU\ket{\psi_1}$ where $E$ is a single or two-qubit logical Pauli error and $U$ is the desired gate. We denote the malignant logical error~$E$ present at the output of the circuit by~$\text{mal}_{E}$. In what follows we will be interested in obtaining estimates of the probability that the event $\mathrm{mal}_{E}$ occurs for the CNOT, Hadamard and $T$ gate. 

We use Monte-Carlo sampling in order to determine the probability of each malignant event given an underlying physical depolarizing model. Given~$N$ simulations of the logical gate $G$ at a physical error rate~$p$, we track the number of malignant faults~$a_E(\epsilon)$ of each error type~$E$, and estimate the probability of a given logical fault as~$\text{Pr}[\text{mal}_E | G,  p] = a_E/N$. The estimate of $\text{Pr}[\text{mal}_E | G, \ p]$ improves as the number of iterations~$N$ increases by reducing the standard deviation. Using a least-squares fitting to determine the error probability as a function of input depolarizing error rate, we can determine the \emph{pseudo-threshold} for each of the logical operations for our error-correcting code. For a level-1 exRec encoding the logical gate~$G$, we define the pseudo-threshold as the crossing point ~$p = p_G^{(1)}(p)$, where $p_G^{(1)}(p)=\sum_{E_i} \text{Pr}[\text{mal}_{E_i} | G, p] $ for all possible logical Pauli errors~$E_i$ for a given logical gate~$G$. Intuitively, the pseudo-threshold corresponds to the error rate below which the logical error rate at level-1 is guaranteed to be lower than the physical error rate. In all previously studied error correction codes, the pseudo-threshold was conjectured to be an upper bound on the asymptotic threshold~\cite{SCCA06, PR12}. In this work we show that this intuitive bound does not necessarily have to hold and that the asymptotic threshold can be much larger than the pseudo-threshold. To our knowledge this is the first exhibition of this type of logical error behaviour and is fundamentally related to the structure of the underlying 105-qubit error correcting code.

At each location of the level-one exRec, errors are introduced following the depolarizing noise model with noise strength $p$. Since the logical gates in question are fault-tolerant, a logical fault can only occur if a sequence of failures occur at the physical level. Namely, we can upper bound the failure probability for each logical fault~$E$ as follows:
\begin{align}
\text{Pr}[\text{mal}_E^{(1)}| G,p] \le \sum_{k = \ceil{\frac{d^*}{2}}}^{L_G} c(k) p^k  =: \Gamma_G^{(1)},
\label{eq:Gamma1}
\end{align}
where the coefficients $c(k)$ are positive integers that parametrize the number of possible weight-$k$ errors that can lead to a logical fault, $L_G$~is the total number of circuit locations in the logical gate~$G$, and $d^*$~characterizes the minimal distance of a given logical gate (that is $\ceil{d^*/2}$~is the minimum weight error that must occur to produce a logical fault). For example in the 105-qubit code, the logical CNOT gate has $d^* = 9$, while the Hadamard and $T$~logical gates have~$d^* = 3$ since they sacrifice some of the distance of the code due to the fact that they are not globally transversal. As was shown in \cite{PR12}, the polynomial $\Gamma^{(1)}(p)$ is monotone non-decreasing making its construction straightforward with the role of upper bounding the logical error probabilities of all the logical operations~$G$ at level-1.

We can then generalize this notion to the level~$l$ concatenation level, where each of the physical locations are replaced by logical~exRec locations of the $(l-1)$~level. Taking the worst case error rate for the $(l-1)$~logical components, the error rate of logical gates at the $l$-th concatenation level can be bounded as follows:
\begin{align}
\text{Pr}[\text{mal}_E^{(l)}| G,p] \le \sum_{k = \ceil{\frac{d^*}{2}}}^{L_G} c(k) \left(\Gamma_G^{(l-1)}\right)^k  =: \Gamma_G^{(l)},
\label{eq:Gammal}
\end{align}
where the polynomials given by the coefficients~$c(k)$ remain the same as the logical gate is composed of the same operations, just replacing physical locations with logical exRecs from the previous concatenation level.

Finally, we generalize a claim of Ref.~\cite{PR12} required to show the suppression of errors for level-2 and higher concatenation levels when below the fault-tolerance threshold~$p_{th}$. Importantly, there exists a $p_{th}$ such that the upper bound on the level-2 logical error probability will be lower than that of level-1, that is $\Gamma_G^{(2)} \le \e \Gamma_G^{(1)}$, and the following will hold for all concatenation levels~$m \ge 2$:
\begin{align}
\text{Pr}[\text{mal}_E^{(m)}| G,p] \le \Gamma_G^{(m)} \le \epsilon^{\ceil{\frac{d^*}{2}}^{m-2}+1} \Gamma_G^{(1)},
\label{eq:asymptoticthresh}
\end{align}
that is the error rate is exponentially suppressed below the crossing point of~$\Gamma_G^{(1)}$ and~$\Gamma_G^{(2)}$, thus providing a lower bound for the asymptotic threshold for the logical gate~$G$. The proof in full generality is presented in Supplementary Material~\ref{app:thresholdlowerbound}.

\section{Concatenated 105-qubit thresholds}
At the level-1 encoding, the logical gate exhibiting the lowest pseudo-threshold is the Hadamard gate~$H$. Due to the complexity of the individual logical Hadamard gates arising on each of the 15-qubit codeblocks, many errors propagating from the different individual gate locations could lead to logical faults on that codeblock. The predominant error occurs when two codeblocks contain a logical fault. The logical error that occurs with the highest probability~$\text{Pr}[\text{mal}_{E} | H, p]$ is a logical~$X$. This can be understood from the sensitivity of the circuit encoding the Hadamard gate to input $Z$ errors from the LEC which have a high probability of leading to a logical error. If any of the input $Z$ errors land on the target qubit of the CNOT gates in the Hadamard encoding circuit, they will propagate to the physical Hadamard gate on the fourth qubit (see Figure~\ref{fig:HadCircuit}) resulting in a logical $X$ error. 

Unlike the logical Hadamard, the leading level-1 logical error for both CNOT and $T$  arise from logical $Z$ errors rather than $X$ errors. This stems from the asymmetry in stabilizer generators of the 15-qubit code resulting in an increased protection against $X$ errors. Due to the transversality of the logical~CNOT gate in both codes and since there are fewer ways for errors to propagate in the implementation of the logical~$T$, these gates have a better pseudo-threshold relative to logical~$H$.

In order to lower bound the level-1 pseudo-threshold, the probability of all logical error types are summed for each of the logical gates and bounded as in Eq.~\ref{eq:Gamma1}. The resulting polynomials are compared to the input physical error rate and their crossing point determines the pseudo-threshold (see Fig.~\ref{fig:PseudoAndLevelThreeThresholds} in the Supplementary Material). The resulting values are presented in Table~\ref{tab:Pseudo-and-asymptotic}.

\begin{table}
\begin{centering}
\begin{tabular}{|c|c|c|}
\hline 
 & Pseudo-Threshold & Asymptotic threshold\tabularnewline
\hline 
\hline 
CNOT gate & $\left(2.11\pm0.02\right)\times10^{-3}$ & $\left(1.95\pm0.01\right)\times10^{-3}$\tabularnewline
\hline 
$T$ gate & $\left(4.89\pm0.11\right)\times10^{-4}$ & $\left(1.58\pm0.02\right)\times10^{-3}$\tabularnewline
\hline 
Hadamard gate & $\left(4.47\pm0.29\right)\times10^{-5}$ & $\left(1.28\pm0.02\right)\times10^{-3}$\tabularnewline
\hline 
\hline
{\bf 105-qubit} & $\mathbf{\left(4.47\pm0.29\right)\times10^{-5}}$ & $\mathbf{\left(1.28\pm0.02\right)\times10^{-3}}$\tabularnewline
\hline 
{\bf 23-qubit Golay} & $\mathbf{\left(1.73\right)\times10^{-3}}$ & $\mathbf{\left(1.32\right)\times10^{-3}}$\tabularnewline
\hline 
\end{tabular}
\par\end{centering}
\caption{\label{tab:Pseudo-and-asymptotic}Lower bounds for the pseudo and asymptotic threshold results for the Hadamard, $T$ gate and CNOT gates. The Hadamard asymptotic-threshold is larger than its pseudo-threshold resulting from the double protection of the CNOT gates as seen by the high CNOT pseudo-threshold. In bold, the overall thresholds for the 105-qubit and 23-qubit codes are compared.}
\end{table}

It is important to observe that the CNOT pseudo-threshold is nearly two orders of magnitude larger than the Hadamard pseudo-threshold. Furthermore, all other operations in our circuits (resting qubits, measurement in the $X$ and $Z$ basis and state preparations) are upper bounded by level one polynomials that have larger pseudo-thresholds than CNOT.
The dominant set of errors leading to logical faults in the level-1 Hadamard gate is a result of input errors from the LEC as well as failures in the CNOT gates within the 15-qubit Hadamard codeblocks. These components are composed of only memory, CNOT, $X$ and~$Z$ basis state preparation and measurement locations. Since the level-1 logical error probability of these gates will be much smaller in the level-2 Hadamard exRec, detrimental faults will be much less likely to occur. Hence, there will be error rates~$p$ above the pseudo-threshold~$p_{1,H}$ such that the level-2 error polynomials characterizing the logical error rate will be below the level-1 bounding polynomial, 

\begin{align}
\Gamma^{(2)}(p)\leq \Gamma^{(1)}(p), \ \forall \ p  \le p_{2,H},
\label{eq:AsymptoticCond}
\end{align}
where $p_{2,H} > p_{1,H} $. The error rate~$p_{2,H}$ is the threshold rate below which all level-2 logical gates have a lower error rate compared to the level-1 logical error rate. As shown in Ref.~\cite{PR12} and argued in the previous section, the value~$p_{2,H}$ serves as a lower bound for the asymptotic threshold~$p_{th}$. 

\begin{figure}
\centering
\begin{subfigure}{0.4\textwidth}
\includegraphics[width=\textwidth]{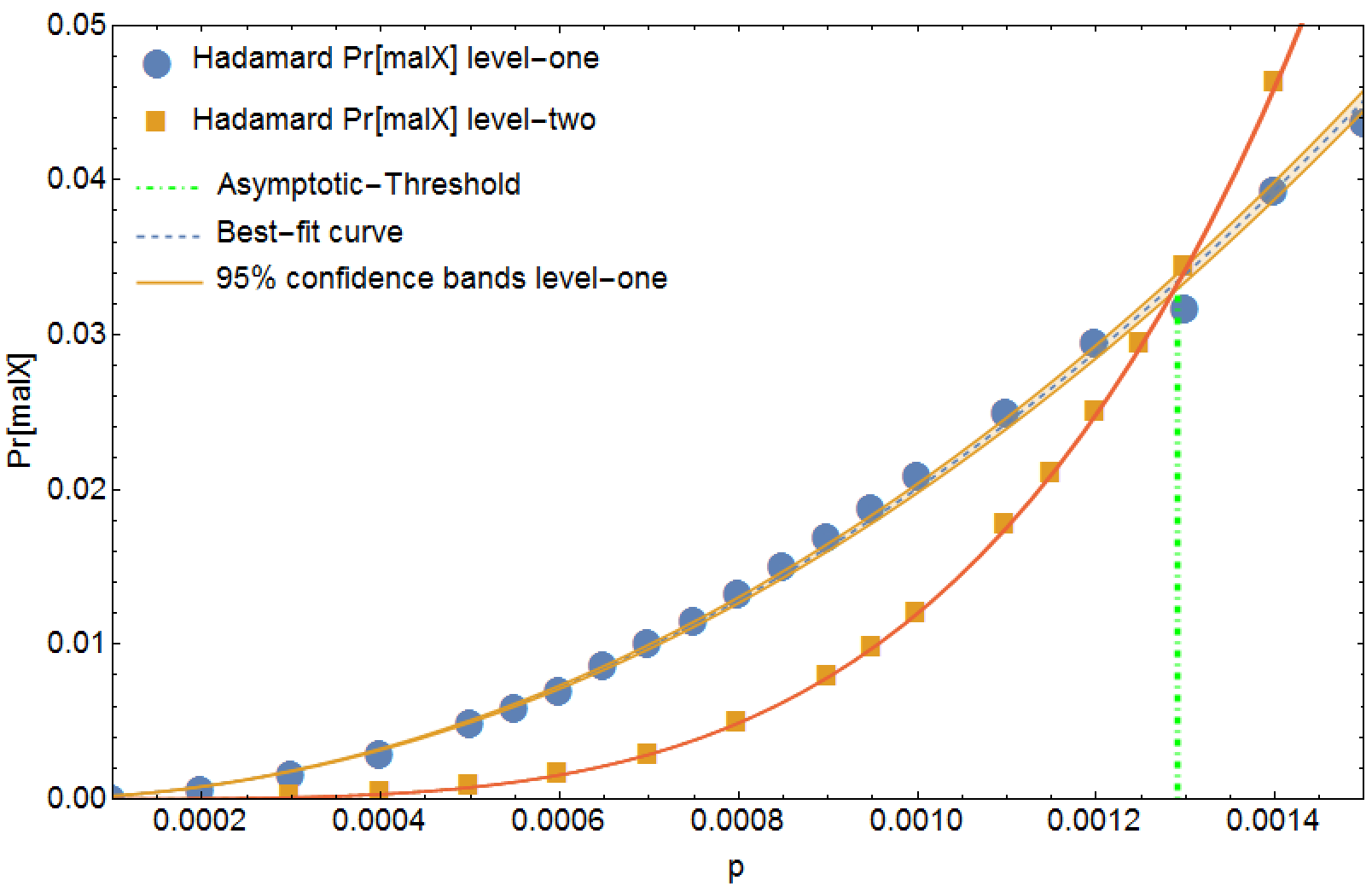}
\caption{}
\label{fig:AsymptoticHadamard}
\end{subfigure}
\begin{subfigure}{0.4\textwidth}
\includegraphics[width=\textwidth]{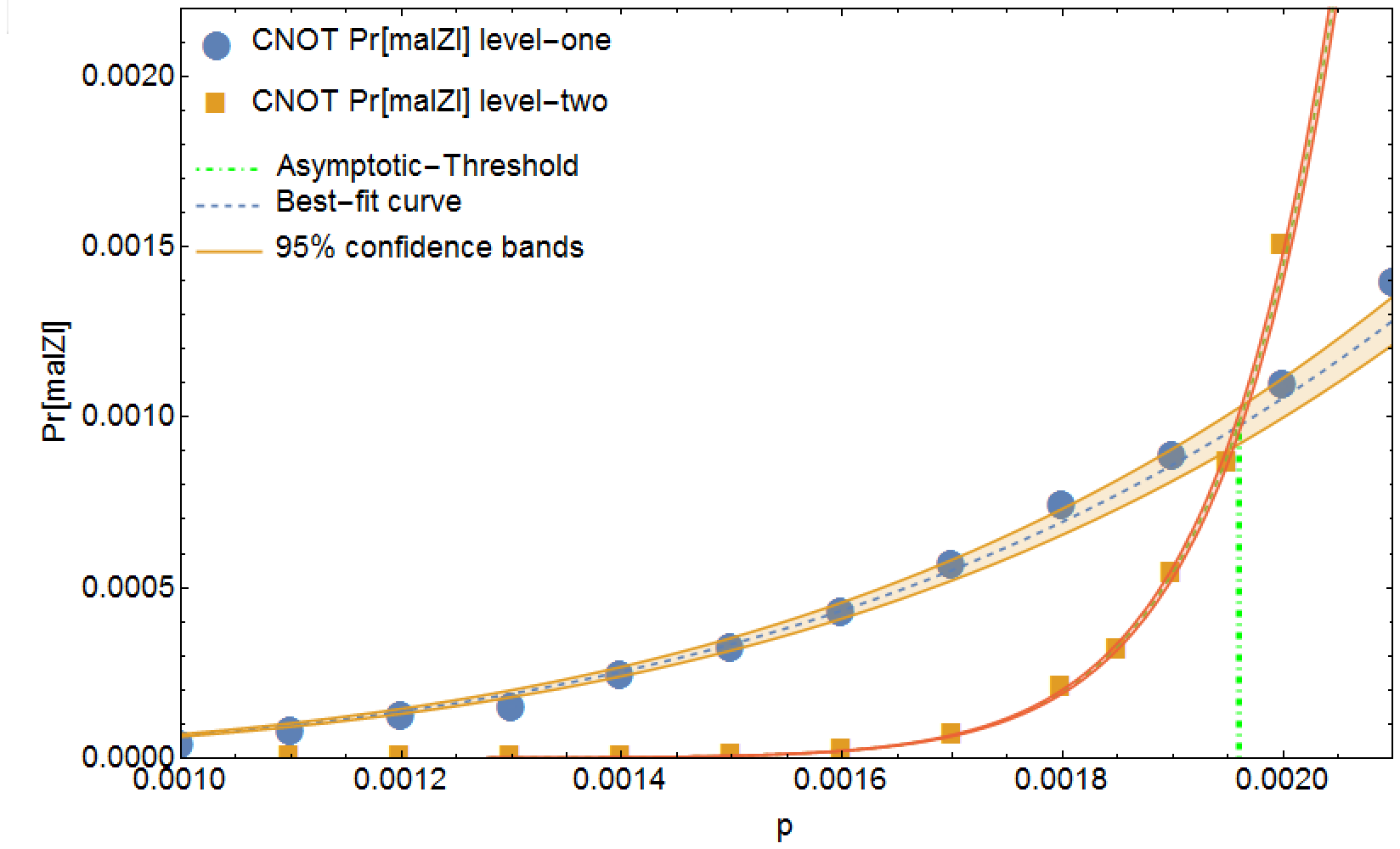}
\caption{}
\label{fig:AsymptoticCNOT}
\end{subfigure}
\caption{Probability of logical error as function of physical error rate for the level-1 and level-2 logical \subref{fig:AsymptoticHadamard}~Hadamard and \subref{fig:AsymptoticCNOT}~CNOT. The crossing point of the fitted curve allows for the determination of a lower bound for the asymptotic threshold for each of the logical gates. The CNOT gate exhibits a much lower logical error rate than the Hadamard at the first level.}
\label{fig:AsymptoticThresh}
\end{figure}

In previous studies of asymptotic thresholds for the Golay and 7-qubit CSS codes, the CNOT exRec provided a lower bound on the threshold value since it contained the largest amount of locations relative to all the other gates in the universal gate set~\cite{AGP06,PR12,CDT09}. Since the CNOT exRec is itself composed entirely of gates that are transversal, as the error rate approaches the pseudo-threshold value, certain malignant events (for example, the probability of getting a logical $ZI$ error at the output of the CNOT circuit, as can be seen in Fig.~\ref{fig:LevelOneCNOT}) become more likely to occur than the level-zero probabilities determined from the depolarizing noise model. Recall that the pseudo-threshold was conjectured to be an upper bound on the asymptotic threshold value. However, it is the CNOT locations that are the leading contributors to logical errors. Consequently, the pseudo-threshold of the CNOT gate, as opposed to the $H$~and~$T$ gates, will be the limiting factor to the asymptotic threshold. As argued above, this will give rise to reduced logical error rates of the~$H$ and~$T$ gates at the second level of concatenation, and using Eq.~\ref{eq:asymptoticthresh}, a lower bound for the asymptotic threshold~$p_{th}$ can be determined. The plots in Fig.~\ref{fig:AsymptoticThresh} illustrate the level-1 and level-2 polynomials upper bounding the logical error rates at the first and second level for the Hadamard and CNOT gate circuits (see Fig.~\ref{fig:PseudoAndLevelThreeThresholds} for the corresponding $T$~gate plots). As expected, the CNOT exRec contains a lower asymptotic threshold value given by $\left(1.95\pm0.01\right)\times10^{-3}$. The Hadamard exRec limits the threshold value of the 105-qubit code to be $\left(1.28\pm0.02\right)\times10^{-3}$. Interestingly, the level-two polynomials satisfy~Eq.~\ref{eq:asymptoticthresh} for error rates nearly 30 times larger than their corresponding level-one polynomials. This is a distinctive feature of the 105-qubit concatenated scheme and clearly demonstrates the impact of having an exRec primarily composed of gates which are transversal in both codes with much larger pseudo-threshold rates. The asymptotic threshold derived for the 105-qubit code compares favourably to the $[[23,1,7]]$~Golay code studied under the same depolarizing error model and metric for gate failures under malignant set counting~\cite{PR12}. This scheme does not require magic state distillation in order to achieve fault-tolerance and may lead to reduced overhead~\cite{FMMC12}. Determining the resource overhead remains an interesting open problem.

\section{Conclusion}
In this work, we established the first rigorous lower bound on the asymptotic threshold for the concatenated 105-qubit code. We show that the pseudo-threshold value of $\left(4.47\pm0.29\right)\times10^{-5}$ arising from the $H$ gate is significantly improved at higher levels of concatenation yielding a lower bound on the asymptotic threshold value of $\left(1.28\pm0.02\right)\times10^{-3}$. The increase in asymptotic threshold is primarily due to the relatively high threshold of the logical CNOT gate.  We believe that this non-traditional behaviour of logical error probabilities at higher concatenation levels is an interesting property of the studied scheme and points to an interesting direction for future error correction research. Due to the high concentration of CNOT gates for the purposes of error detection, we believe that tailoring codes to correct for logical errors in encoded CNOT gates at the expense of perhaps noisy single qubit gates would allow for higher asymptotic thresholds for concatenated codes.

\section{Acknowledgements}
T.~J.~would like to acknowledge the support of NSERC and the Vanier-Banting Secretariat through the Vanier~CGS. This work was supported by CIFAR, NSERC, and Industry~Canada. C. C. would like to thank Hemant Katiyar for useful discussions.

\bibliographystyle{ieeetr}
\bibliography{bibtex_jochym}

\newpage
\appendix
\renewcommand\thefigure{S\arabic{figure}} 
\renewcommand\theequation{S\arabic{equation}} 
\setcounter{figure}{0}
\setcounter{equation}{0}      

\section{Error type and detailed threshold analysis}
\label{app:ErrorType}
In this section we provide a more detailed threshold analysis for the CNOT, $H$ and $T$ gates. Furthermore, we provide details on the contributions from different error types. 

\begin{figure}[h]
\centering
\includegraphics{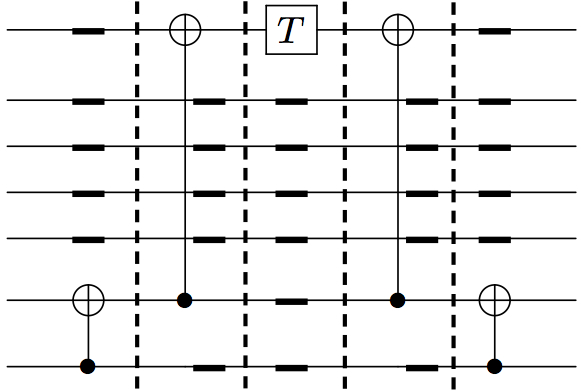}
\caption{Logical $T$ gate circuit for the 105-qubit concatenated code. The $T$ gate is applied transversally on the 15-qubit codeblocks. A single error at any location in the above circuit can result in a single error on multiple codeblocks which will be corrected at the 15-qubit level. Consequently, the logical $T$ gate circuit is fault-tolerant at the 105-qubit level.}
\label{fig:LogicalTgateCircuit}
\end{figure}

Fig.~\ref{fig:LogicalTgateCircuit} illustrates the $T$ gate circuit construction for the 105-qubit concatenated code. Notice that compared to the Hadamard circuit constrcution (Fig.~\ref{fig:HadCircuit}), there are much fewer locations where errors can propagate leading to a logical fault on multiple codeblocks. Since the 105-qubit code is more efficient at correcting $X$ errors, and $X$ errors propagating through a physical $T$ gate location transforms as $TXT^{\dagger}=X(I + iZ)/\sqrt{2}$ (leading to a $Z$ error contribution), we expect the probability of obtaining a logical $Z$ error at the output of the $T$ gate circuit to be much higher than the probability of obtaining a logical $X$ error. In fact, our simulations showed that the level-one logical $X$ error probability could be upper bounded by $10^{-7}$ for all considered error rates. 

Given the large number of locations in our encoding circuits, it is computationally intractable to compute exact upper bounds on the probability of failure for the different subcomponents of an exRec as was done in Ref.~\cite{PR12}. Instead we use a Monte-Carlo sampling technique to estimate the probability of failure of a particular exRec. At the first level of concatenation, we insert $X$, $Y$ or $Z$ Pauli errors at each physical locations of the exRec with a probability governed by the depolarizing noise model. Once all the error locations are fixed, we propagate the errors through the exRec and verify the output for a logical error. This procedure consists of one simulation. Recall that for $N$ simulations of a logical gate $G$ at a physical error rate $p$, the probability of a given logical fault for an error of type $E$ is estimated by
\begin{align}
\text{Pr}[\text{mal}_E | G,  p] = \frac{a_E}{N}
\label{eq:level1boundingAppendix}
\end{align} 
where $a_E$ is the number of malignant faults for an error of type $E$. In choosing $N=10^7$ simulations, we obtained statistical error deviations ranging between $10^{-7}$ to $10^{-5}$ which we felt were adequate for our estimates in the threshold values (see the uncertainty relations obtained in Table~\ref{tab:Pseudo-and-asymptotic}). 

\begin{figure}[h]
\centering
\begin{align*}
\Qcircuit @C=1.6em @R=1em {
& \gate{\text{EC}} & \gate{G_1} & \gate{\text{EC}} & \gate{G_2} & \gate{\text{EC}} & \qw 
\gategroup{1}{2}{1}{4}{0.7em}{--}
\gategroup{1}{4}{1}{6}{1.5em}{--}
}
\end{align*}
\caption{An example of shared EC's between two consecutive level-one exRecs}
\label{fig:OverlappingECs}
\end{figure}
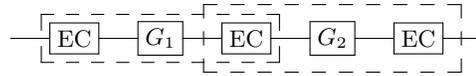

\subsection{Error type analysis}

In computing the probability of obtaining a logical error $E$ at the second level of concatenation for an error rate $p$ and gate $G$ ($\text{Pr}[\text{mal}_E^{(2)}| G,p]$), each level-one exRecs in the level-two circuit was treated as a physical independent location with a redefined noise model given by the polynomials $\Gamma_{G,E}^{(1)}(p)$. For example, a level-one CNOT gate in a level-two simulation would be treated as a physical CNOT gate. A two-qubit Pauli error would be inserted with a probability upper bounded by the polynomials obtained in Eq.~\ref{eq:Gammal} with $l=2$ instead of the probability arising from the depolarizing noise model. To be consistent with Eq.~\ref{eq:AsymptoticCond}, the notation is chosen such that $\Gamma_G^{(l)}$ is the upper bounding polynomial at level-$l$ for all error types $E$ whereas $\Gamma_{G,E}^{(l)}$ is the upper bounding polynomial for the particular error type $E$.

As can be seen in Fig.~\ref{fig:OverlappingECs}, a level-two simulation will typically contain many level-one exRecs with overlapping EC’s and so it is not entirely correct to treat them independently. If two level-one exRecs share EC’s, then the rectangle that precedes the other one is replaced with a faulty gate only if it is still incorrect after the shared EC’s have been removed. As was shown in Ref.~\cite{PR12}, we must calculate ~\ref{eq:level1boundingAppendix} for both complete and incomplete exRecs where one or more TEC’s have been removed and take the polynomial that bounds all cases. See Fig.~\ref{fig:malIZECRem} for an example. We would also like to point out that for single qubit gates, exRec's without a TEC always had a lower probability of obtaining a logical fault (for any error type) compared to the case where the TEC was kept. This can be understood from the fact that the TEC adds more locations and hence more ways for errors to be introduced at the output of the circuit.  

\begin{figure}[h]
\centering
\begin{subfigure}{0.4\textwidth}
\includegraphics[width=\textwidth]{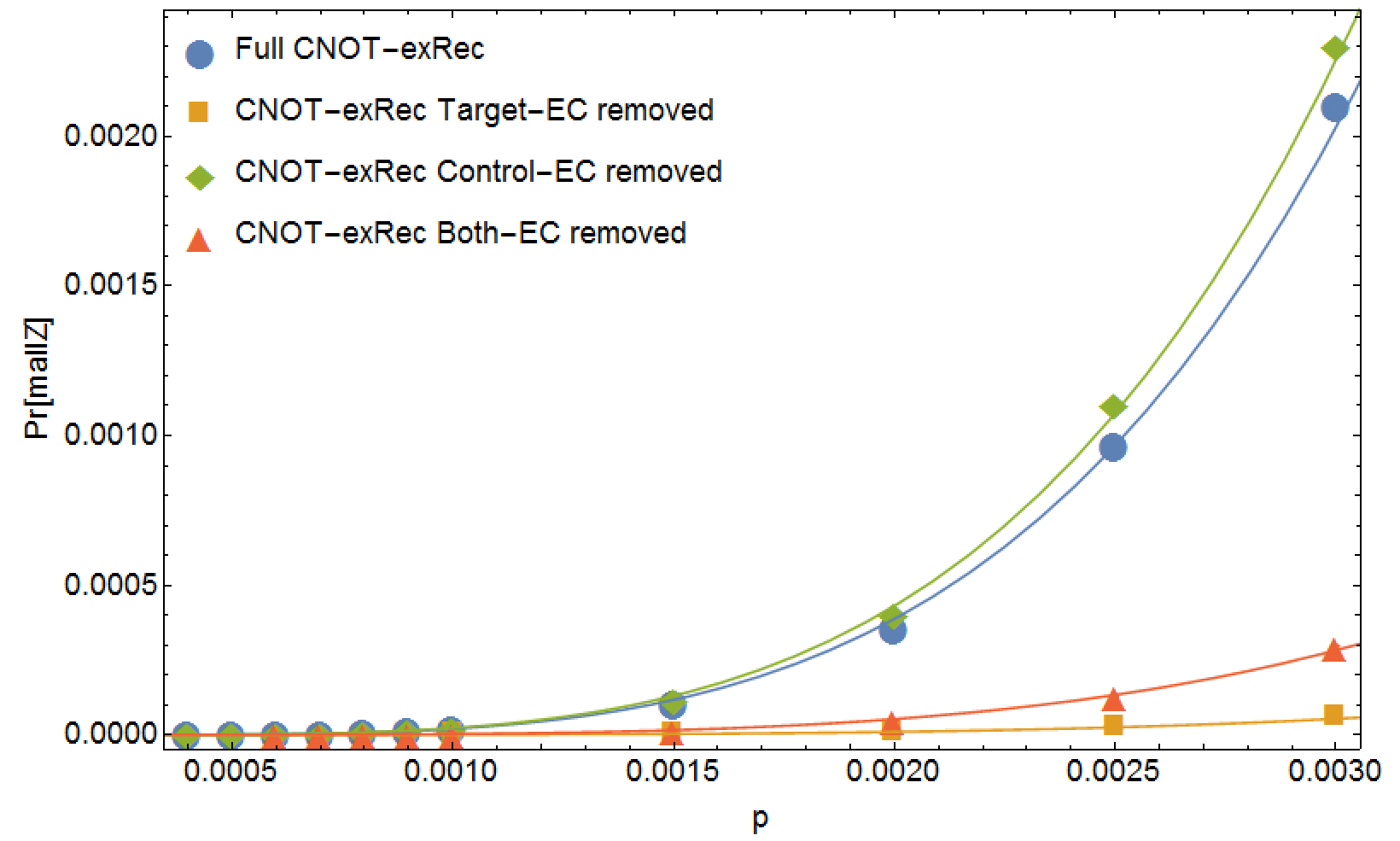}
\caption{}
\label{fig:malIZECRem}
\end{subfigure}
\begin{subfigure}{0.4\textwidth}
\includegraphics[width=\textwidth]{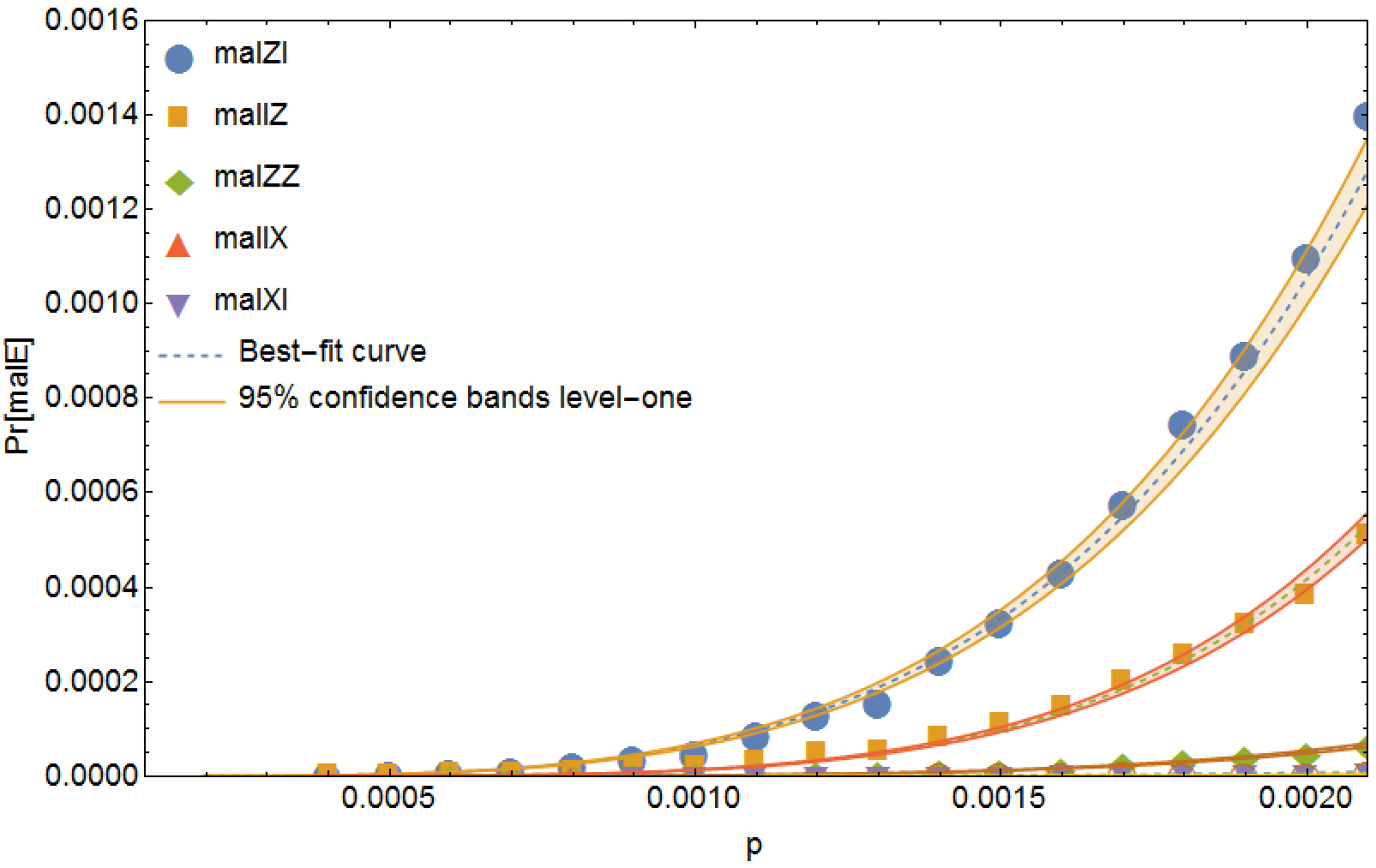}
\caption{}
\label{fig:LevelOneCNOT}
\end{subfigure}
\caption{\subref{fig:malIZECRem}~Polynomials upper bounding the event $\text{mal}_{IZ}$ for either the full level-one CNOT exRec or the level-one CNOT exRec with one or both TEC's removed. The polynomial upper bounding the event $\text{mal}_{IZ}$ will upper bound all the curves in the above figure. \subref{fig:LevelOneCNOT}~Polynomials upper bounding the level-one CNOT exRec for the different logical error types.}
\label{fig:CNOTPseudo}
\end{figure}

The polynomials in Fig.~\ref{fig:LevelOneCNOT} upper bound the probability of obtaining a logical error at the first level of concatenation of the CNOT exRec. Each curve corresponds to a different error type (error types that are not displayed occur with a probability less than $10^{-7}$ for all sampled physical error rates). Note that the upper bounds on logical $Z$ malignant events are significantly higher than their $X$ counterpart. As mentioned in the main text, this is primarily due to the fact that the 15-qubit Reed-Muller code offers better protection against $X$ errors.

\begin{figure}[h]
\centering
\begin{subfigure}{0.4\textwidth}
\includegraphics[width=\textwidth]{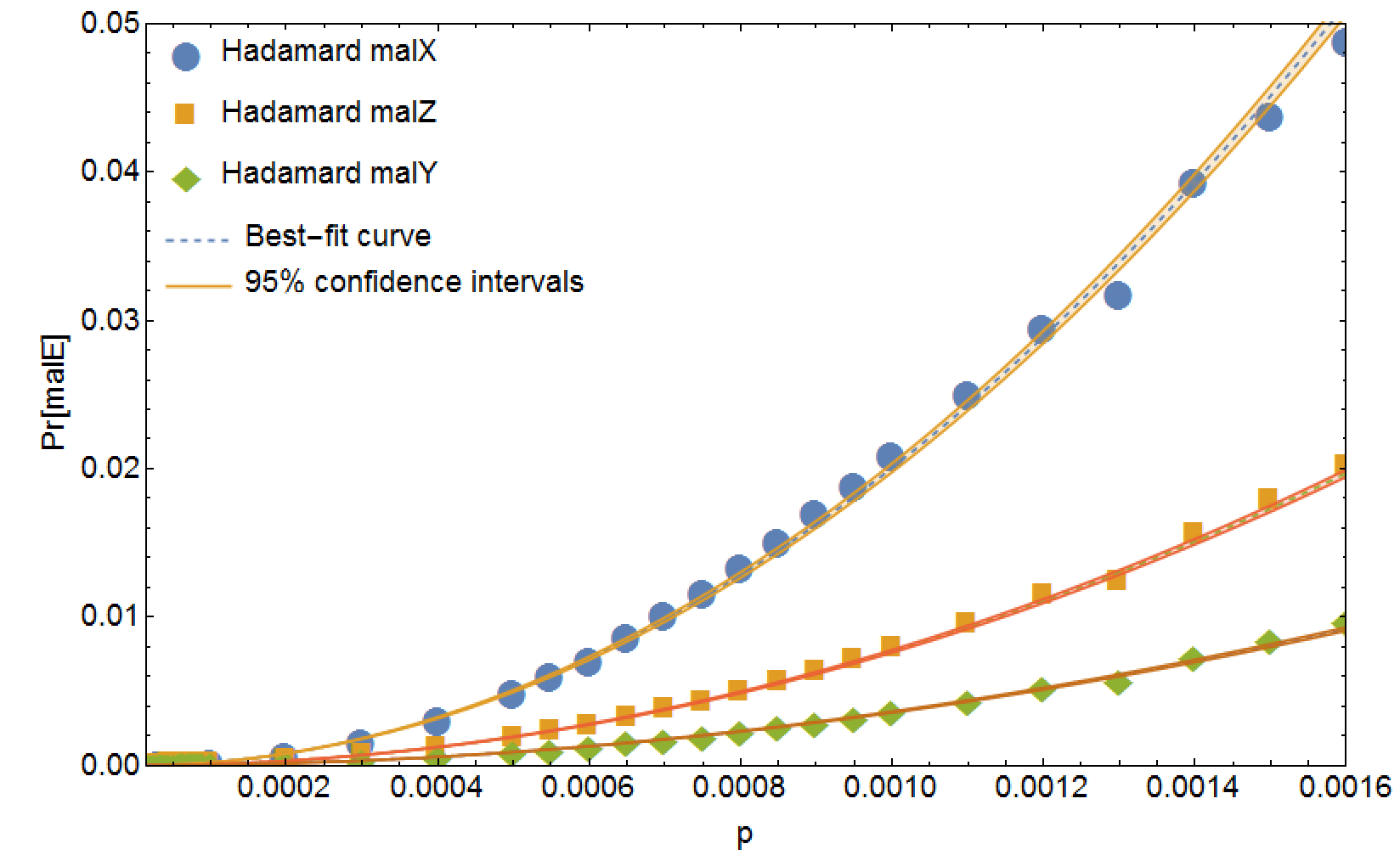}
\caption{}
\label{fig:LevelOneHad}
\end{subfigure}
\begin{subfigure}{0.4\textwidth}
\includegraphics[width=\textwidth]{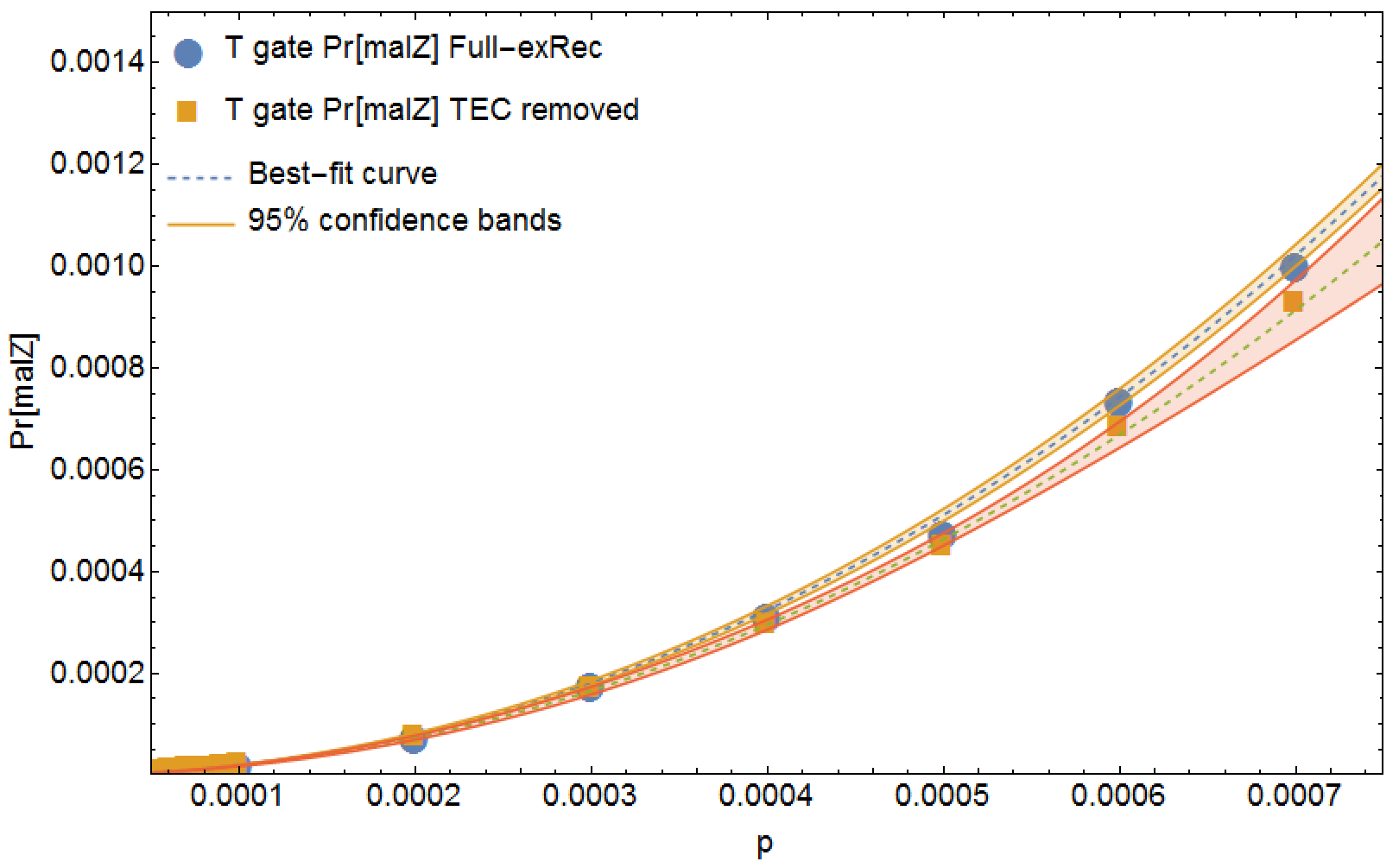}
\caption{}
\label{fig:LevelOneTgate}
\end{subfigure}
\caption{\subref{fig:LevelOneHad}~Polynomials upper bounding the events $\text{mal}_{X}$, $\text{mal}_{Z}$ and $\text{mal}_{Y}$ for the level-one Hadamard exRec. Input $Z$ errors are most likely to result in a logical $X$ error on a 15-qubit codeblock which explains why the event $\text{mal}_{X}$ is most likely to occur. \subref{fig:LevelOneTgate}~Polynomials upper bounding the event $\text{mal}_{Z}$ for the level-one T gate exRec. Note that the logical error probabilities for the event $\text{mal}_{X}$ and $\text{mal}_{Y}$ are too small to be displayed.}
\label{fig:SingleGatesPseudo}
\end{figure}

The polynomials of Fig.~\ref{fig:LevelOneHad} upper bound the probability of obtaining a logical $X$, $Z$ or $Y$ error for the level-one Hadamard exRec. As was explained in section \ref{Concatenated 105-qubit scheme} of the main text, the circuit encoding the logical $H$ on the 15-qubit codeblocks is very sensitive to input $Z$ errors. Any $Z$ error propagating throught the target qubit of the CNOT gates (prior to the application of the physical $H$ on qubit 4, see Fig.~\ref{fig:HadCircuit}) will result in a logical $X$ error on the 15-qubit codeblock. The latter is the main reason for a larger upper bound on the event $\text{mal}_{X}$ compared to the event $\text{mal}_{Z}$ even though the 15-qubit Reed-Muller code provides better protection against $X$ errors.

\subsection{Threshold analysis}

The pseudo-threshold values for a gate $G$ given in Table~\ref{tab:Pseudo-and-asymptotic} were obtained from the crossing point between the physical error rate $p$ and the curves $p_G^{(1)}(p)=\sum_{E_i} \text{Pr}[\text{mal}_{E_i} | G, p] $ for all possible logical Pauli errors~$E_i$. The plots on the left column of Fig.~\ref{fig:PseudoAndLevelThreeThresholds} illustrates the crossing point for the logical CNOT, Hadamard and $T$ gate. The CNOT gate has the largest pseudo-threshold value of $\left(2.11\pm0.02\right)\times10^{-3}$ due to the double protection from the CSS 7-qubit Steane code and the 15-qubit Reed-Muller code. On the other hand, the Hadamard gate has the lowest pseudo-threshold value of $\left(4.47\pm0.29\right)\times10^{-5}$ due to the sensitivity of the encoding circuit on the 15-qubit codeblocks to input $Z$ errors. 

Following Eq.~\ref{eq:AsymptoticCond}, a lower bound for the asymptotic threshold value for a particular gate $G$ is given by the the intersection between the polynomials upper bounding the probability of obtaining a logical error $E$ at the first and second level of concatenation ($\Gamma_{G,E}^{(1)}$ and $\Gamma_{G,E}^{(2)}$). Note that the error type $E$ in the asymptotic threshold calculation is chosen such that the intersection between $\Gamma_{G,E}^{(1)}$ and $\Gamma_{G,E}^{(2)}$ occurs at the smallest physical error rate. For the CNOT gate (Fig.~\ref{fig:b}) this is given by~$E=ZI$, for the Hadamard gate (Fig.~\ref{fig:d}) it is~$E=X$ and for the $T$ gate (Fig.~\ref{fig:f}) it is~$E=Z$.

An interesting feature can be observed from the plots on the right column of Fig.~\ref{fig:PseudoAndLevelThreeThresholds}. Notice that the polynomial upper bounding the event $\text{mal}_{ZI} $ at the third level of concatenation for the logical CNOT gate $\Gamma_{CNOT,ZI}^{(3)}(p)$ intersects $\Gamma_{CNOT,ZI}^{(1)}(p)$ at the asymptotic threshold value $\left(1.95\pm0.01\right)\times10^{-3}$. The reason is that for higher error rates than the asymptotic threshold value, the level-two CNOT exRecs (in the level three simulation) are more likley to fail than the level-one CNOT exRecs (in the level two simulation). Consequently, there is a higher probability of obtaining a logical fault at the output of the CNOT exRec. However, for the logical $H$ and $T$ gate exRecs, $\Gamma_{G,E}^{(3)}(p)$ intersects $\Gamma_{G,E}^{(1)}(p)$ at an error rate which is larger than the asymptotic threshold value for these particular gates ($\left(1.28\pm0.02\right)\times10^{-3}$ for $H$ and $\left(1.58\pm0.01\right)\times10^{-3}$ for $T$). Consider the logical Hadamard gate (the following argument applies equally well to the $T$ gate). For error rates $p$ that are between the $H$ and CNOT asymptotic threshold values, $\left(1.28\pm0.02\right)\times10^{-3}\leq p\leq\left(1.95\pm0.01\right)\times10^{-3}$, the level-two Hadamard exRecs in the level-three simulation will be more likely to fail than at the previous level of concatenation. However, this will be compensated by all of the level-two CNOT exRecs in the level-three simulation which will be less likely to fail than at the previous level (since $p$ is below the CNOT asymptotic threshold value). Above the error rate where $\Gamma_{H,X}^{(3)}(p)$ intersects $\Gamma_{H,X}^{(1)}(p)$ ($p=1.44\times10^{-3}$), the level-two Hadamard exRecs will be noisy enough such that the probability of obtaining a logical $X$ error will be larger than at the previous level. Therefore, by considering the crossing points of the logical error rates for higher concatenation levels, a better lower-bound for the asymptotic threshold can be established. However, in order to fairly compare the performance of the concatenated scheme with the Golay code~\cite{PR12}, we emphasized the lower bound obtained from the crossing point of the first and second concatenation levels.

\section{Lower bound on asymptotic threshold}
\label{app:thresholdlowerbound}
\setcounter{equation}{1} 

We review how we arrived at Eq.~\ref{eq:asymptoticthresh} and how this result leads to a lower bound on the asymptotic threshold for the code in question. We prove a more general result for the exponential suppression of error rates as a function of concatenation levels given the presence of a crossing point of the upper bounding polynomials of the error rate at consecutive concatenation levels. 

\begin{lemma}
Suppose the error rate of a logical gate~$G$ at the $l$-th concatenation level can be upper bounded as follows:
\begin{align}
\text{Pr}[\text{mal}_E^{(l)}| G,p] \le \Gamma_G^{(l)} = \sum_{k = \ceil{\frac{d^*}{2}}}^{L_G} c(k) \left(\Gamma_G^{(l-1)}\right)^k.
\label{eq:polyexpansion}
\end{align}
If the upper bounding error polynomial satisfies the following $\Gamma^{(l)} \le \epsilon \Gamma^{(l-1)}$ for $0 \le \epsilon \le 1$, then the following holds:
\begin{align}
\text{Pr}[\text{mal}_E^{(m)}| G,p] \le \Gamma_G^{(m)} \le \epsilon^{\sum_{r=0}^{m-l} \ceil{\frac{d^*}{2}}^r} \Gamma_G^{(l-1)},
\label{eq:expsuppression} 
\end{align}
where $m>l$, and $d^*$ is the minimal distance of the encoded state throughout the logical application of the gate~$G$.
\end{lemma}

\begin{proof}
We shall show this result by induction. Therefore, consider first the case of $m = l+1$. By definition $\text{Pr}[\text{mal}_E^{(l+1)}| G,p]) \le \Gamma_G^{(l+1)}$, for all logical errors~$E$. In order to show the right side of the inequality given in Eq.~\ref{eq:expsuppression} consider the expansion of $\Gamma_G^{(l+1)}$ as a sum over failures of the gates at the $(l)$-th~level, and use the claim that~$\Gamma_G^{(l)} \le \e \Gamma_G^{(l-1)}$.
\begin{align*}
\Gamma_G^{(l+1)} &= \sum_k c(k) \left( \Gamma_G^{(l)} \right)^k \\
&\le \sum_k c(k) \left( \epsilon \Gamma_G^{(l-1)} \right)^k \\
&= \e^{\ceil{\frac{d^*}{2}}} \sum_k c(k) \epsilon^{k - \ceil{\frac{d^*}{2}}}  \left( \Gamma_G^{(l-1)} \right)^k \\
&\le \e^{\ceil{\frac{d^*}{2}}} \sum_k c(k) \left( \Gamma_G^{(l-1)} \right)^k \\
&= \e^{\ceil{\frac{d^*}{2}}} \Gamma_G^{(l)} \\
&\le \e^{\ceil{\frac{d^*}{2}}+1} \Gamma_G^{(l-1)}
\end{align*}
We used the fact that all of the $c(k)$~coefficients in the expansion are positive and due to the fault-tolerance of the logical gates, errors of order smaller than~$\ceil{d^*/2}$ are correctable and therefore $c(k) = 0 \ \forall \  k < \ceil{d^*/2}$.

To complete the proof, we assume the induction hypothesis for level~$m$ and show for level~$(m+1)$:
\begin{align*}
\Gamma_G^{(m+1)} &= \sum_k c(k) \left( \Gamma_G^{(m)} \right)^k \\
&\le \sum_k c(k) \left( \epsilon^{\sum_{r=0}^{m-l} \ceil{\frac{d^*}{2}}^r} \Gamma_G^{(l-1)} \right)^k \\
&\le \epsilon^{ \ceil{\frac{d^*}{2}} \sum_{r=0}^{m-l} \ceil{\frac{d^*}{2}}^r} \sum_k c(k) \left( \Gamma_G^{(l-1)} \right)^k \\
&= \epsilon^{  \sum_{r=1}^{m+1-l} \ceil{\frac{d^*}{2}}^r} \sum_k c(k) \left( \Gamma_G^{(l-1)} \right)^k \\
&= \epsilon^{  \sum_{r=1}^{m+1-l} \ceil{\frac{d^*}{2}}^r}  \Gamma_G^{(l)}  \\
&\le \epsilon^{  \sum_{r=0}^{m+1-l} \ceil{\frac{d^*}{2}}^r} \Gamma_G^{(l-1)} ,
\end{align*}
thus completing the induction proof.
\end{proof}

It should be noted that the shift in the crossing point for different concatenation levels in the logical $H$ and $T$~gate (Figs.~\ref{fig:d} and~\ref{fig:f}) may at first glance violate the assumption that the polynomial coefficients~$c(k)$ are the same at all levels. However, one of the assumptions of the polynomials were that the logical error rate of all locations at the previous level have the same error rate, and thus contribute equally in a potential error chain. The fact that CNOT is in fact less noisy than other gates in the regime between the $H$ (and $T$) pseudo-threshold and asymptotic CNOT threshold means that certain error chains are further suppressed and as such the logical error rate is lower than the worst case bound set by the polynomials. The CNOT crossing points (Fig.~\ref{fig:b}) are uniform across all levels, indicating that the true logical error rate is very close to the worst-case bound.

\begin{figure*}[htbp]
\begin{center}
\begin{subfigure}{0.4\textwidth}
\includegraphics[width = \textwidth]{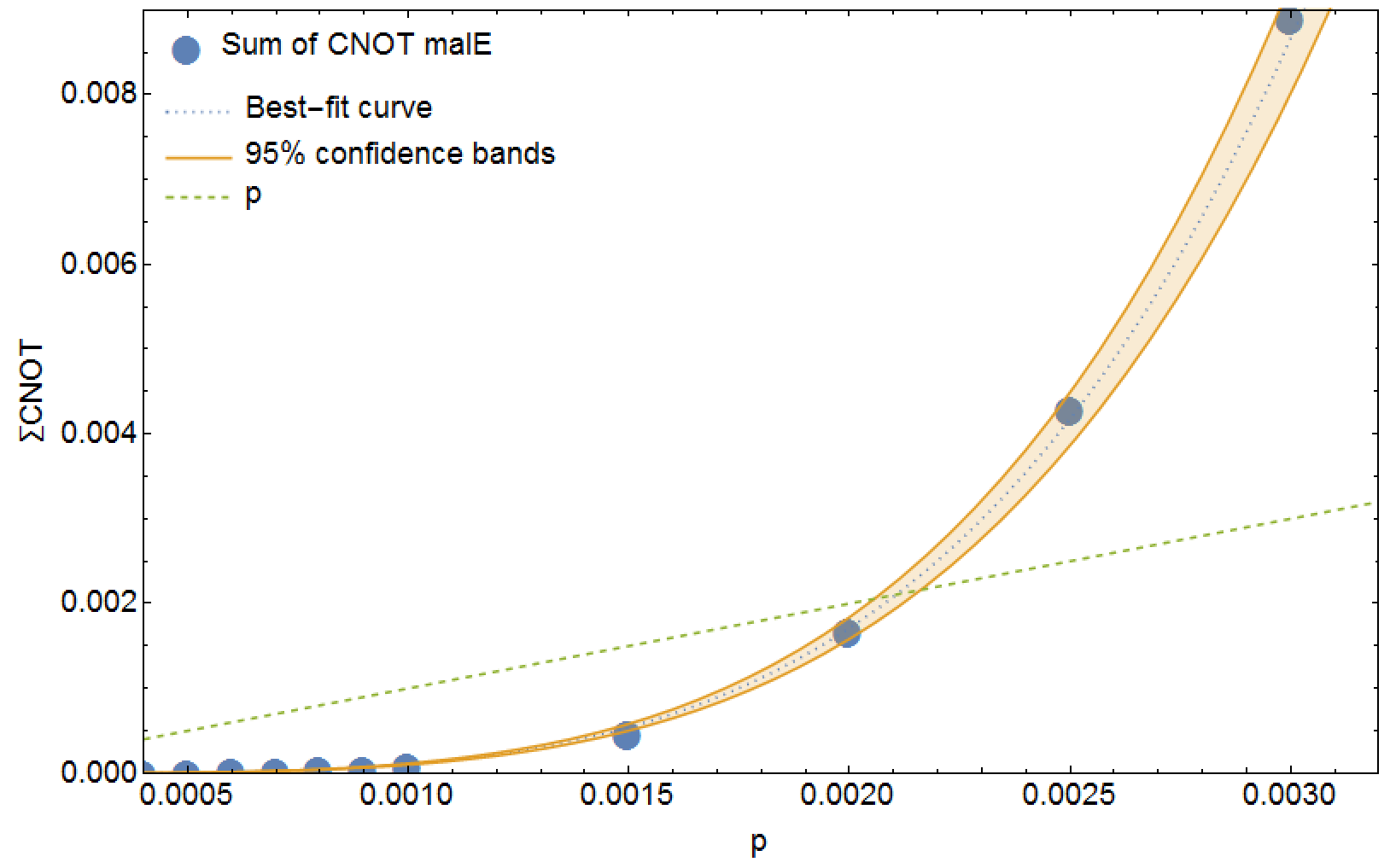}
\caption{}
\label{fig:a}
\end{subfigure}
\begin{subfigure}{0.4\textwidth}
\includegraphics[width =\textwidth]{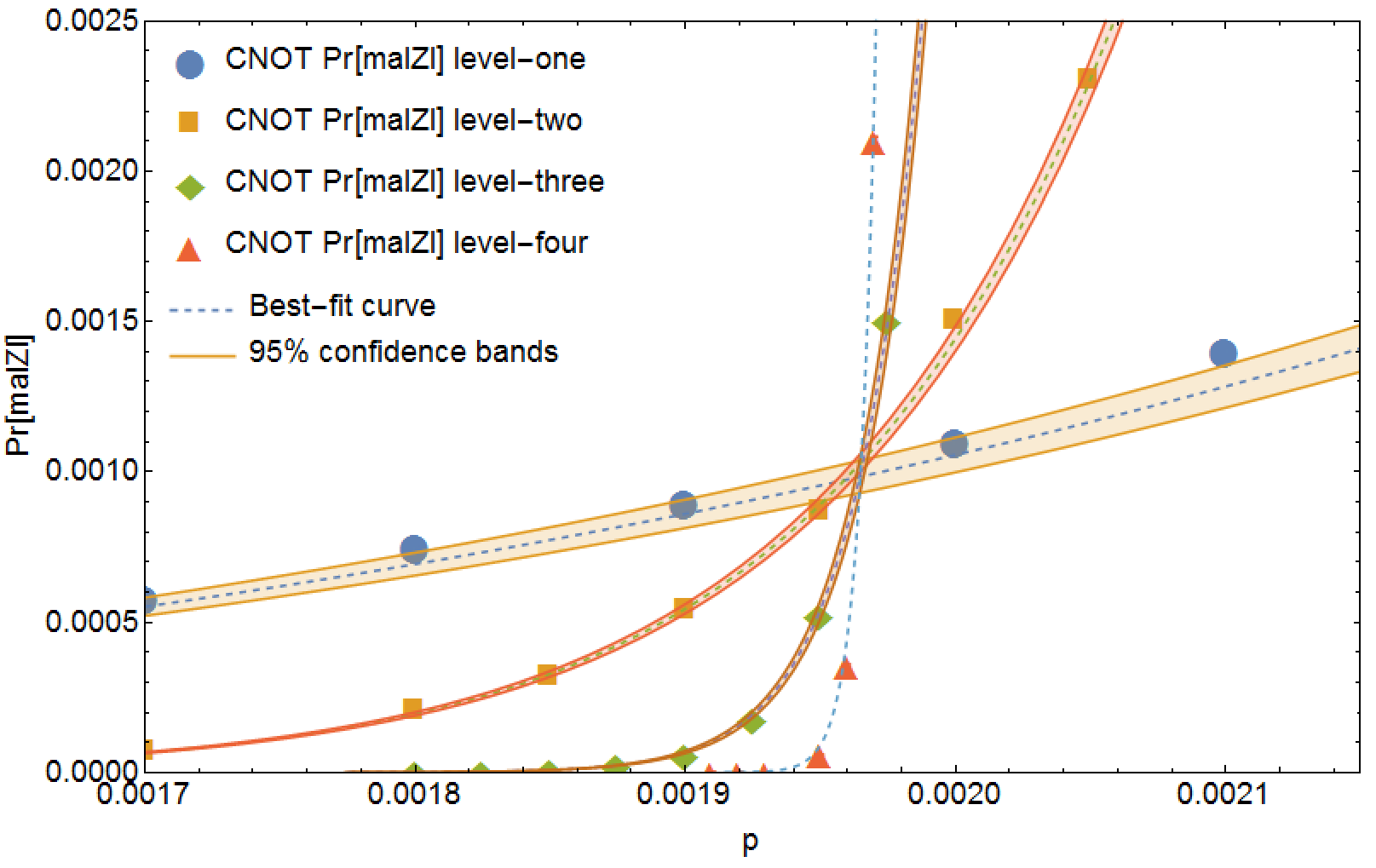}
\caption{}
\label{fig:b}
\end{subfigure}
\begin{subfigure}{0.42\textwidth}
\includegraphics[width = \textwidth]{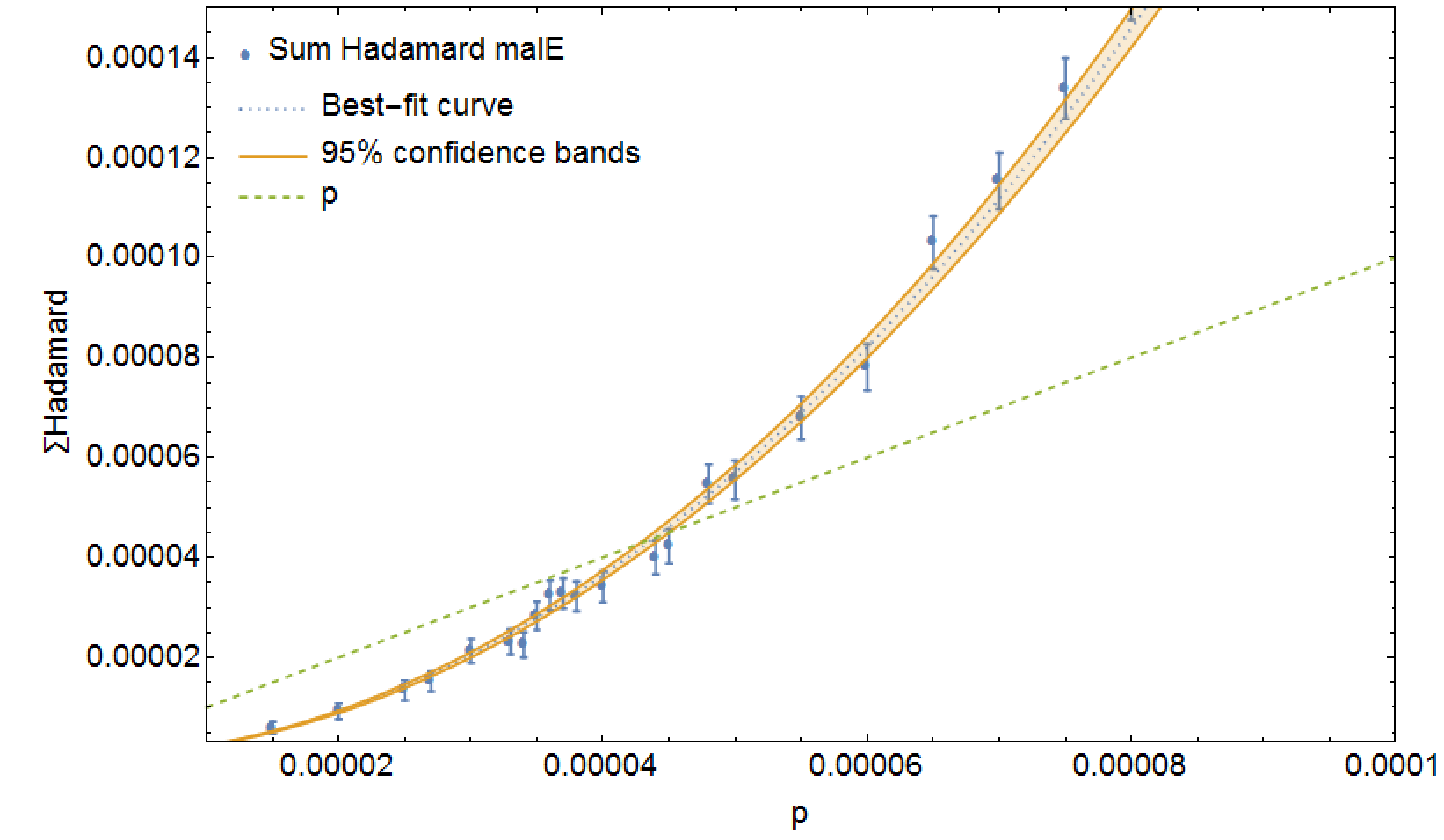}
\caption{}
\label{fig:c}
\end{subfigure}
\begin{subfigure}{0.4\textwidth}
\includegraphics[width = \textwidth]{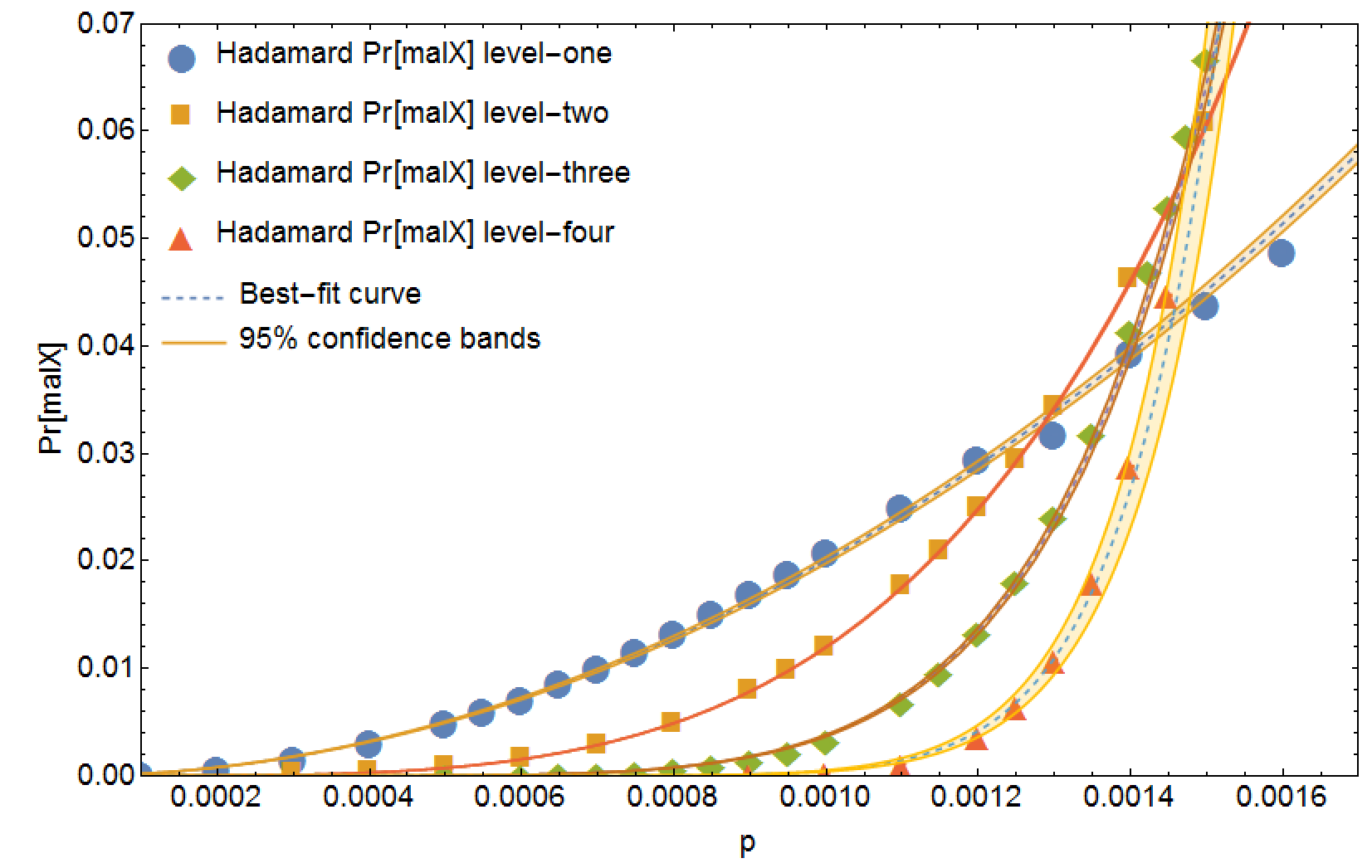}
\caption{}
\label{fig:d}
\end{subfigure}
\begin{subfigure}{0.42\textwidth}
\includegraphics[width = \textwidth]{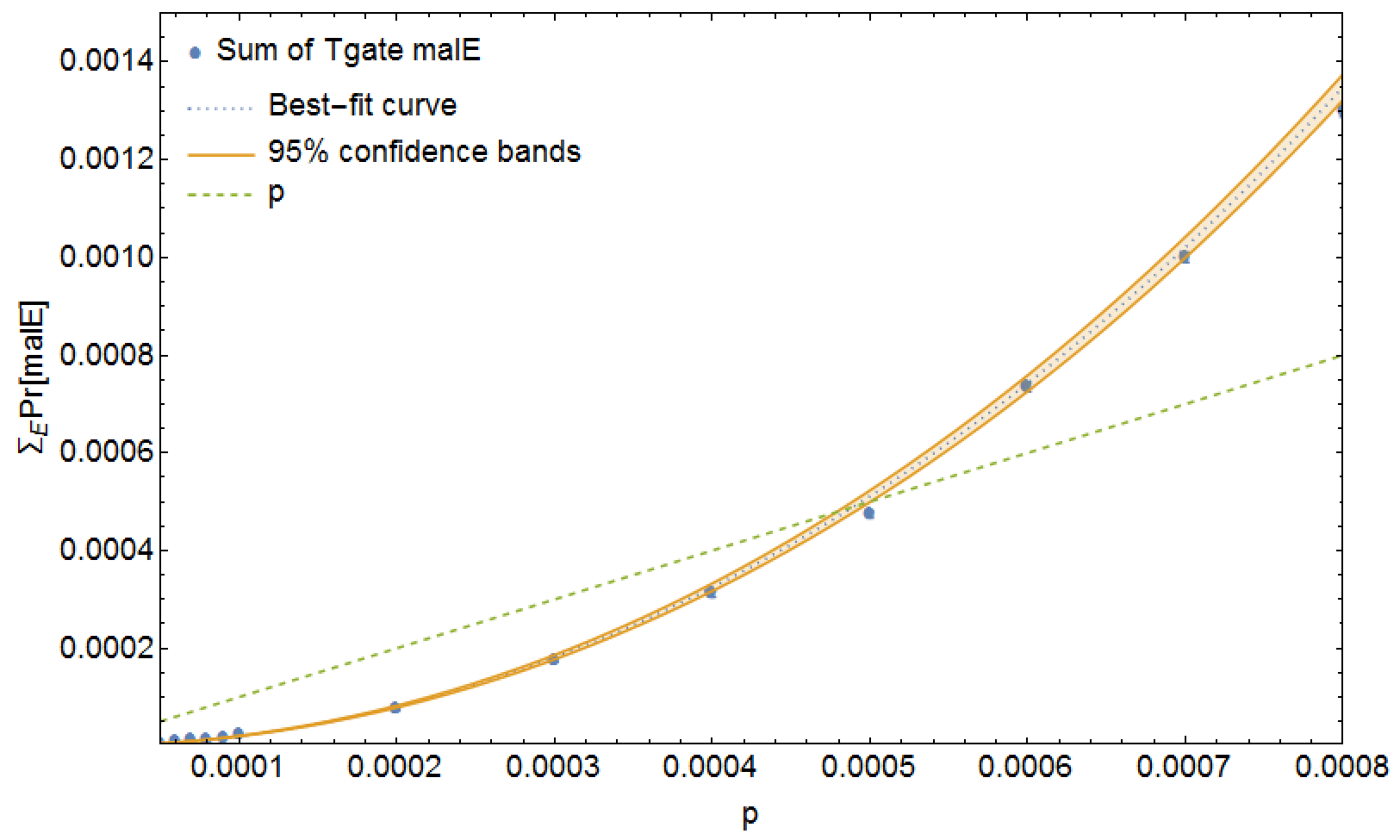}
\caption{}
\label{fig:e}
\end{subfigure}
\begin{subfigure}{0.4\textwidth}
\includegraphics[width = \textwidth]{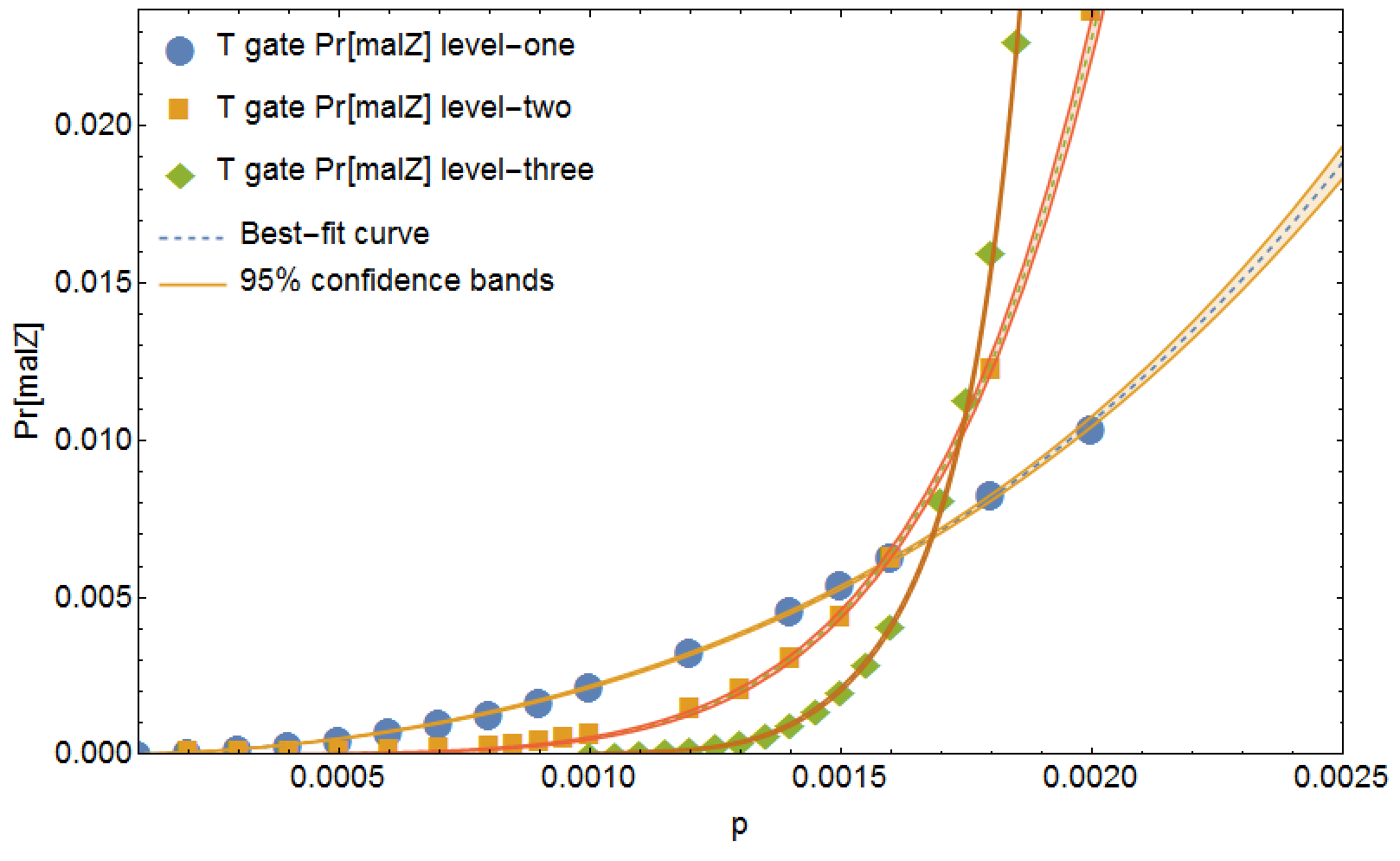}
\caption{}
\label{fig:f}
\end{subfigure}
\caption{The plots on the left column illustrate the probability of logical error as function of physical error rate for logical \subref{fig:a}~CNOT, \subref{fig:c}~Hadamard and \subref{fig:e}~$T$~gate. The crossing point of the fitted curve allows for the determination of the level-1 pseudo-threshold for each of the logical gates. The CNOT pseudo-threshold is the largest among all three gates due to the double protection of the 7-qubit and 15-qubit code. The plots on the right column illustrate the polynomials upper bounding the probability of obtaining a logical error $E$ for the first, second and third level of concatenation. The crossing point between the level-one and level-two polynomials determine the asymptotic threshold for the gate under consideration. For the logical CNOT gate~\subref{fig:b}, it is the event $\text{mal}_{ZI}$ which limits the threshold value. For the logical gate $H$~\subref{fig:d}, $\text{mal}_{X}$ limits the threshold value. Lastly, for the logical $T$~gate~\subref{fig:f}, $\text{mal}_{Z}$ limits the threshold value.}
\label{fig:PseudoAndLevelThreeThresholds}
\end{center}
\end{figure*}

\end{document}